\documentclass[11pt]{amsart}
\usepackage[utf8]{inputenc}

\usepackage{defs}
\usepackage{tikz}
\usepackage{subcaption}
\usetikzlibrary{positioning}
\usepackage{macros}
\usepackage{comment}
\usepackage{ulem}
\title{Dynamic Tolling for Inducing Socially Optimal Traffic Loads}
\author{Chinmay Maheshwari\({}^{*,1}\), Kshitij Kulkarni\({}^{*,1}\), Manxi Wu\({}^{1}\), and S. Shankar Sastry\({}^{1}\)}
\thanks{* Equal contribution}
\thanks{\(^{1}\) The authors are with the Department of Electrical Engineering and Computer Sciences, University of California Berkeley, USA. Research supported in part by the  U.S. Office of Naval Research MURI grant N00014-16-1-2710 and by the NSF Partnership for International Research and Education Excellence (PIRE): Science of Design for Societal Scale Cyber-Physical Systems, National Science Foundation award number OISE-1743772.}
\date{October 2021}
\begin{document}

\maketitle
\thispagestyle{empty}
\pagestyle{empty}

\begin{abstract}
    How to design tolls that induce socially optimal traffic loads with dynamically arriving travelers who make selfish routing decisions? We propose a two-timescale discrete-time stochastic dynamics that adaptively adjusts the toll prices on a parallel link network while accounting for the updates of traffic loads induced by the incoming and outgoing travelers and their route choices. The updates of loads and tolls in our dynamics have three key features: (i) The total demand of incoming and outgoing travelers is stochastically realized; (ii) Travelers are myopic and selfish in that they choose routes according to a perturbed best response given the current latency and tolls on parallel links; (iii) The update of tolls is at a slower timescale as compared to the the update of loads. We show that the loads and the tolls eventually concentrate in a neighborhood of the fixed point, which corresponds to the socially optimal load and toll price. Moreover, the fixed point load is also a stochastic user equilibrium with respect to the toll price. Our results are useful for traffic authorities to efficiently manage traffic loads in response to the arrival and departure of travelers. 
\end{abstract}
\section{Introduction}\label{sec: Intro}

Efficient traffic routing is an increasingly important problem as congestion aggravates on urban transportation networks. Many cities around the world have implemented congestion pricing (tolling) to regulate the demand of travelers (\cite{latimes2021,nyt2021}). Tolling, when properly set in response to real-time traffic conditions, can effectively alleviate congestion and even induce socially optimal loads \cite{sharon2017real, paccagnan2019incentivizing}. 

One challenge of setting tolls is that travelers continuously arrive at and depart from the network, which creates a dynamic evolvement of traffic loads in network. Additionally, travelers are selfish in that they prefer to choose the routes with the minimum cost - travel time cost and the tolls combined. To design tolling that induces efficient traffic load, we must account for the dynamics of incoming and outgoing traffic demands and the selfish nature of travelers' route choices. In particular, the tolls must be updated in response to the changing traffic conditions.


We propose a discrete-time stochastic dynamics to capture the joint evolution of the loads and tolls in a parallel network. In each time step of the dynamics, non-atomic travelers arrive at the origin of the network, and they make routing decisions according to a perturbed best response based on the travel time cost and toll of each link in that step. Additionally, a fraction of load on each link leaves the network. Both the incoming and outgoing demand are randomly realized, and are identically and independently distributed across steps. Therefore, the discrete-time stochastic dynamics of loads forms a Markov process, which is governed by the total arriving demand, the stochastic user equilibrium, and the load discharge rate. Furthermore, at each time step, a traffic authority adjusts the toll on each link by interpolating between the current toll and a new increment dependent on the marginal cost of travel time given the load at that step. 

In our setting, the dynamics of the toll evolves at a slower time-scale compared to that of the load dynamics. In practice, fast changing tolls are undesirable (\cite{farokhi2015piecewise}). This property ensures that the tolls change very slowly, and thus travelers can view the tolls as static when they make routing decisions at the arrival. 



We show that the loads and tolls in the discrete-time stochastic dynamics asymptotically concentrate in a neighborhood of a unique fixed point with high probability. The fixed point load is \emph{socially optimal} in that it minimizes the total travel time costs when the incoming and outgoing traffic demands reach a \emph{steady state}, and the fixed point toll on each link equals to the marginal cost. That is, with high probability, our dynamic tolling eventually induces the socially optimal loads that accounts for the incoming and outgoing travelers and their selfish routing behavior. Furthermore, we emphasize that our dynamic tolling is \emph{distributed} in that the traffic authority only uses the information of the cost and load on each link to update its toll.  


Our technical approach to proving the main result involves: \emph{(i)} Constructing a continuous-time deterministic dynamical system associated with the two timescale discrete-time stochastic dynamics; \emph{(ii)} Demonstrating that the flow of the continuous-time dynamical system has a unique fixed point that corresponds to the socially optimal load and tolls; \emph{(iii)} Proving that the unique fixed point of the flow of the continuous time dynamical system is globally stable. In particular, we apply the theory of two time-scale stochastic approximation to show that the loads and tolls under the stochastic dynamics concentrates with high probability in the neighborhood of the fixed point of the flow of continuous time dynamics constructed in \emph{(i)} (\cite{borkar2009stochastic}). Additionally, our proof in \emph{(ii)} on the uniqueness and optimality of fixed point of the flow of continuous time dynamics builds on a variational inequality, and extends the analysis of stochastic user equilibrium in static routing games to account for the steady state of the network given the incoming and outgoing demand (\cite{cominetti2012modern}). Furthermore, we show that the continuous time dynamical system is cooperative, and thus its flow must converge to its fixed point (\cite{hirsch1985systems}).

Our model and results contribute to the rich literature on designing tolling mechanisms for inducing socially optimal route loads. Classical literature on static routing games has focused on measuring the inefficiency of selfish routing by bounding the "price of anarchy", and designing marginal cost tolling to induce socially optimal route loads (\cite{christodoulou2005price, roughgarden2002bad, roughgarden2010algorithmic}). In static routing games, optimal tolling does not account for the dynamic arrival and departure of travelers. The computation of optimal tolls relies on knowledge of the entire network structure and the equilibrium route flows, which are challenging to compute. 

Dynamic toll pricing has been studied in a variety of settings to account for the continuous incoming and outgoing traffic demand. The paper \cite{borkar2009cooperative} analyzed a discrete-time stochastic dynamics of of non-atomic travelers, and discussed the impact of tolling on routing strategies. We consider an adaptive adjustment of tolls using marginal toll pricing. This allows us to analyze the long-run outcomes of the joint evolution of the route loads and tolls, and shows that the tolls eventually induce a socially optimal traffic load associated with the steady state of the network. Additionally, we prove that the monotonicity condition  -- the equilibrium routing strategy is monotonic in tolls (which is an assumption in \cite{borkar2009cooperative}) -- holds for any equilibrium routing strategy. 

Moreover, \cite{como2021distributed} proposed a continuous-time dynamical system to study socially optimal tolling when strategic travelers continuously arrive and make selfish routing decisions. In their model, the incoming and outgoing traffic demands are equal so that the aggregate load in the network is a constant. In our setting, both the incoming and outgoing traffic demands are random variables. Therefore, our fixed point analysis needs to account for the total load at the steady state of the network. Moreover, \cite{como2021distributed} assumes that the tolls are adjusted at a faster time scale than their traveler's route preferences, while we assume that the update of tolls is at a slower timescale compared to the change of routing decisions.  

Finally, this article is also related to the literature on learning in routing games, and learning for tolling with unknown network condition. In particular, a variety of algorithms have been proposed to study how travelers learn an equilibrium by repeatedly adjusting their routing decisions based on the observed travel time in the network (e.g. \cite{cominetti2010payoff, krichene2014convergence, kleinberg2009multiplicative}). Additionally, papers \cite{farokhi2015piecewise, poveda2017class, meigs2017learning, wu2020bayesian} have analyzed how the traffic authority adaptively updates the tolls while learning the unknown network condition using crowd-sourced data on traffic load and time costs.



This article is organized as follows: we introduce the dynamic tolling model and the discrete time stochastic dynamics in Sec. \ref{sec: Model}. We present the main result in Sec. \ref{sec: Results}, and numerical examples in Sec. \ref{sec: NumExp}. We conclude our work in Sec. \ref{sec: Conclusion}. We present the key ideas of our prove techniques in the main text, and include all proofs in the appendix. 

\subsection*{Notations}
We denote the set of non-negative real numbers by \(\Rp\). For any natural number \(R\), we succinctly write \([R] = \{1,2,\dots,R\}\). For any vector \(x\in \R^n\), we define \(\diag(x)\in\R^{n\times n}\) to be a diagonal matrix with its diagonals filled with entries of \(x\).

\section{Model}\label{sec: Model}
  Consider a parallel link network with \(\numLinks\) links connecting a single source-destination pair. At each time step $n=1, 2, \dots$, a non-atomic traveler population arrives at the source node, and is routed through links \(\numLinks\) to the destination node. 
  
  At time step $n$, the traffic load (i.e. amount of travellers) on link \(i\in [\numLinks]\) is \(\flowDiscrete_i(n)\). The latency function on any link \(i \in [\numLinks]\) is a function of load on the link: \(\latency_i: \Rp \ra\Rp\). We assume that the latency function is \emph{strictly-increasing} and \emph{convex}, which reflect the congestible nature of links and the fact that the latency increases faster when the load is higher.

  A traffic authority sets toll prices on links, denoted as $\priceDiscrete(n)=\left(\priceDiscrete_i(n)\right)_{i \in \R}$, where $\priceDiscrete_i(n)$ is the toll on link $i$ at step $n$. The \emph{cost function} on link \(i\in[\numLinks]\) at step $n$, denoted $\costPrice_i(\flowDiscrete_i(n), \priceDiscrete_i(n))$, is the sum of the latency function \(\latency_i(\flowDiscrete_i(n))\) and the charged toll price (see Fig. \ref{fig: SchematicFigure}): 
  \begin{align}\label{eq: costAssumption}
      \costPrice_i(\flowDiscrete_i(n), \priceDiscrete_i(n))) \defas \latency_i(\flowDiscrete_i(n)) + \priceDiscrete_i(n)
  \end{align}
  
  \begin{figure}[h]
\centering
\begin{tikzpicture}[main/.style = {draw, circle}] 
\node[main] (1) {$\mc{S}$}; 
\node[main] (2) [right= 5cm of 1] {$\mc{D}$}; 
\draw[->] (1) to [out=70,in=110,looseness=0.8] node[midway,above] {$\costPrice_1(\flowDiscrete_1(n), \priceDiscrete_1(n))$} (2);
\draw[->] (1) to [out=45,in=135,looseness=0.4] node[midway,above] {$\costPrice_2(\flowDiscrete_2(n) ,\priceDiscrete_2(n))$} (2);
\draw[->] (1) to [out=0,in=180,looseness=0.4] node[midway,above] {$\dots$} (2);
\draw[->] (1) to [out=-45,in=225,looseness=0.4] node[midway,below] {$\costPrice_{R-1}(\flowDiscrete_{R-1}(n), \priceDiscrete_{R-1}(n))$} (2);
\draw[->] (1) to [out=-70,in=-110,looseness=0.8] node[midway,below] {$\costPrice_{R}(\flowDiscrete_R(n), \priceDiscrete_R(n))$} (2);
\end{tikzpicture} 
\caption{A \(\numLinks-\)link parallel network with source node \(\mc{S}\), destination node \(\mc{D}\), and cost functions at step \(n\).}
\label{fig: SchematicFigure}
\end{figure}
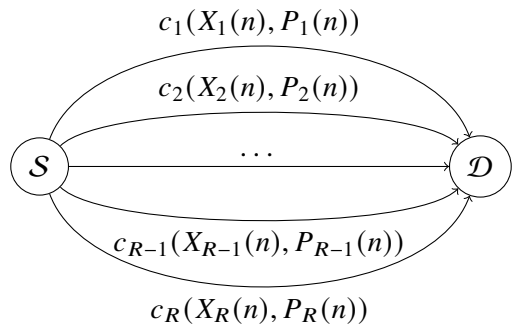

  Furthermore, we define \(\tilde{\costPrice}_i(\flowDiscrete_i(n), \priceDiscrete_i(n), \nu_i(n)) \defas \costPrice_i(\flowDiscrete_i(n) , \priceDiscrete_i(n)) + \nu_i(n)\) to be the randomly realized travel cost experienced by the travelers, where $\nu_i(n)$ is an identically and independently distributed (i.i.d) random variable with zero mean. 
  

  The demand of traffic arriving at the source node at step \(n+1\) is a random variable $\inflow(n+1)$. We assume that  \(\{\inflow(n)\}_{n\in\N}\) are i.i.d with bounded support $[\inflowLower,\inflowUpper]$ and the mean $\lambda$ (i.e. \(\avg[\inflow(n)] = \meanInflow\) for all \(n\)). At step \(n+1\), travelers make routing decisions $\flowIn(n+1) = \left(\flowIn_i(n+1)\right)_{i \in [\numLinks]}$ based on the latest cost vector at step $n$, where $\flowIn_i(n+1)$ is the demand of travelers who choose link $i$ at step \(n+1\). We assume that $\flowIn(n+1)$ is a perturbed best response defined as follows:  
  
 \begin{definition}\label{assm: InflowAssumption}(\textbf{Perturbed Best Response Strategy})
  At any step \(n+1\), the routing strategy $\flowIn_i(n+1)$ is \emph{perturbed best response} if for all \(i\in[\numLinks]\),
  \begin{align}\label{eq: InflowEquation}
     \flowIn_i(n+1) \defas \frac{\exp(-\sueParameter \costPrice_i(\flowDiscrete_i(n), \priceDiscrete_i(n)))}{\sum_{j\in[R]}\exp(-\sueParameter \costPrice_j(\flowDiscrete_j(n), \priceDiscrete_j(n)))}\inflow(n),  
 \end{align}
 where $\beta \in [0, \infty)$. 
  \end{definition}

The perturbed best response strategy reflects the myopic nature of travelers route choices. In particular, $\sueParameter$ is a dispersion parameter that governs the relative weight of link costs in making routing decisions. If $\beta \uparrow \infty$, $\flowIn_i(n+1)$ is a best response strategy in that travelers only take links with the minimum cost in $\flowIn_i(n+1)$. On the other hand, if $\beta\downarrow 0$, then $\flowIn_i(n+1)$ assigns the arrival demand uniformly across all links.  
 


Furthermore, the proportion of load discharged from link \(i\in[\numLinks]\) at step \(n+1\) is given by the random variable \(\outflow_i(n+1) \in (0,1)\). We assume that \(\{\outflow_i(n)\}_{i\in[\numLinks],n\in\N}\) are i.i.d. with bounded support \([\outflowLower,\outflowUpper]\) and mean \(\meanOutflow\) (i.e. \(\avg[\outflow_i(n)] = \mu\) for all \(i\in[\numLinks]\)). Thus the load discharged from link \(i\in[\numLinks]\) at step \(n+1\) is given by \begin{align}\label{eq:outflow}
    \flowOut_i(n+1) \defas \flowDiscrete_i(n)\outflow_i(n+1).
\end{align}

In each step $n$, the load on each link is updated as follows: 
\begin{align}\label{eq:x_original}
    \flowDiscrete_i(n+1) &= \flowDiscrete_i(n) +  \flowIn_i(n+1) - \flowOut_i(n+1) .
\end{align}
We note that the stochasticity of the load update arises from the randomness in the incoming load $\flowIn(n+1)$ and the outgoing load $\flowOut(n+1)$. We define 
\begin{align}\label{eq: hVectorFieldDef}
\vectorFieldX_i(x_i,p_i) \defas \frac{\meanInflow}{\meanOutflow}\frac{\exp(-\sueParameter \costPrice_i(x_i,p_i))}{\sum_{j\in[R]}\exp(-\sueParameter \costPrice_j(x_j,p_j))}. 
\end{align}
where $x, p \in \mathbb{R}^{\numLinks}$, and $\meanInflow$ (resp. $\meanOutflow$) is the mean of incoming (resp. outgoing) load respectively.
Using \eqref{eq: InflowEquation}, \eqref{eq:outflow}, and \eqref{eq: hVectorFieldDef}, we can re-write \eqref{eq:x_original} as follows: 
\begin{equation}\tag{Update-$\flowDiscrete$}\label{eq: FlowDynamics}
\begin{aligned}
    \flowDiscrete_i(n+1)
    &=(1-\meanOutflow) \flowDiscrete_i(n) + \meanOutflow \vectorFieldX_i(\flowDiscrete_i(n), \priceDiscrete_i(n)) + \meanOutflow\martingale_i(n+1),
\end{aligned}
\end{equation}
where
\begin{equation}
\begin{aligned}\label{eq: MartingaleDefinition}
\martingale_i(n+1)  & \defas   \vectorFieldX_i(\flowDiscrete_i(n), \priceDiscrete_i(n))(\inflow(n+1) - \meanInflow) \\ &\quad - \flowDiscrete_i(n)\lr{\outflow_i(n+1)- \meanOutflow}. 
\end{aligned}
\end{equation}


The central authority updates the toll vector $\priceDiscrete(n) \in \mathbb{R}^{\numLinks}$ at each step $n$ as follows: 
\begin{align}\tag{Update-$\priceDiscrete$}\label{eq: PricingDynamics}
\priceDiscrete_i(n+1) = (1-a)\priceDiscrete_i(n) + a \flowDiscrete_i(n)\frac{d \costEdge_i(\flowDiscrete_i(n))}{d x}
\end{align}
where \( i \in [\numLinks], n\in \N\) and \(a\in [0, 1]\) is the step size.
That is, the updated toll is an interpolation between the current toll and the \emph{marginal cost} of the link given the current load. We note that the toll is updated in a distributed manner in that $\priceDiscrete_i(n+1)$ only depends on the load and cost on link $i$. 


The updates of $(\flowDiscrete(n), \priceDiscrete(n))$ are jointly governed by the stochastic updates in \eqref{eq: FlowDynamics} and \eqref{eq: PricingDynamics}. We assume that tolls are updated at a slower timescale compared with the load. That is, \(a \ll \meanOutflow\), where $a$ (resp. $\meanOutflow$) is the step size in \eqref{eq: PricingDynamics} (resp. \eqref{eq: FlowDynamics}). 


\color{black}
\section{Main Results}\label{sec: Results}
In Sec. \ref{subsec:fixed_point}, we present a continuous-time deterministic dynamical system that is associated with the discrete-time updates \eqref{eq: FlowDynamics} and \eqref{eq: PricingDynamics}. We show that the flow of the continuous-time dynamical system has a unique fixed point that corresponds to the perturbed socially optimal load (refer Definition \ref{def:perturbed_opt}) and the socially optimal tolling. 
In Sec. \ref{ssec: ConvergenceDiscrete}, we apply the two timescale stochastic approximation theory to show that $\left(\flowDiscrete(n), P(n)\right)$ in the
discrete-time stochastic updates concentrate on a neighborhood of the fixed point of the flow of the continuous-time dynamical system. Therefore, our dynamical tolling eventually induces a {perturbed socially optimal load vector} with high probability. 

\subsection{Continuous-time dynamical system}\label{subsec:fixed_point}
We first introduce a \emph{deterministic} continuous-time dynamical system that {corresponds to} \eqref{eq: FlowDynamics} -- \eqref{eq: PricingDynamics}. The time evolution in continuous-time dynamical system is denoted by \(t\) and it is related with the discrete-time step \(n\) as \(t=an\), where \(a\) is the stepsize in \eqref{eq: PricingDynamics}. We define $x(t) \in \mathbb{R}^{\numLinks}$ as the load vector and $p(t) \in \mathbb{R}^{\numLinks}$ as the toll vector at time $t \in [0, \infty)$ 
We also define \(\stepRatio \defas \frac{\stepPrice}{\meanOutflow} \), where $\meanOutflow$ (resp. $\stepPrice$) is the stepsize of the discrete-time load update (resp. toll update). Since the toll update occurs at a slower timescale compared to the load update (i.e. $\stepPrice \ll \meanOutflow$), we have $\stepRatio \ll 1$. 

The continuous-time dynamical system is as follows: 
\begin{equation}\label{eq: singularPerturbedODE}
    \begin{aligned}
    \dot{\flowContinuous}_i(t) &= \frac{1}{\stepRatio}\lr{\vectorFieldX_i(\flowContinuous(t),\priceContinuous(t)) - \flowContinuous_i(t)}, \\ 
    \dot{\priceContinuous}_i(t) &= -\priceContinuous_i(t) + \flowContinuous_i(t)
    \frac{d \ell_i(\flowContinuous_i(t))}{d x}, \quad \forall \ t \geq 0.
    \end{aligned}
\end{equation}

%
We introduce the following two definitions: 
\begin{definition}\label{def:generalized_SUE}{(\textbf{Stochastic user equilibrium})}
 For any fixed \(p\in \R^{\numLinks}\), a load vector $\fixedPointFlow(\priceContinuous)$\footnote{We explictly state the dependence on \(p,\beta\) in order relate the definition to later results.} is the stochastic user equilibrium corresponding to 
the toll vector $\priceContinuous$ and demand $\frac{\meanInflow}{\meanOutflow}$ if for all \(i\in[\numLinks]\):
 \begin{align}\label{eq: xEquilibrium}
    \fixedPointFlow_i(\priceContinuous) = 
    \frac{\meanInflow}{\meanOutflow}\frac{\exp(-\sueParameter \costPrice_i(\fixedPointFlow_i(\priceContinuous),\priceContinuous_i))}{\sum_{j\in[R]}\exp(-\sueParameter \costPrice_j(\fixedPointFlow_j(\priceContinuous),\priceContinuous_j))}.
\end{align}
\end{definition}
Given the stochastic user equilibrium in \eqref{eq: xEquilibrium}, all travelers with total demand of $\lambda/\mu$ make routing decisions in the network according to a perturbed best response given the latency functions on links and tolls $p$. We note that demand $\lambda/\mu$ is the expected value of the total demand in network at step $n$ when $n \uparrow \infty$. This is because in each step $n$, the expected value of the total demand in network is $\sum_{m=1}^{n}(1-\mu)^{n-m}\lambda$, where \((1-\mu)^{n-m}\lambda\) is the expected incoming demand in step $m$ that remains in the network in step $n$. Thus, as $n \uparrow \infty$, the expected value of the total demand is $\lambda/\mu$. 

\begin{definition}\label{def:perturbed_opt}(\textbf{Perturbed socially optimal load})
A load vector $\bar{y}^{(\beta)}$ is a perturbed socially optimal load if $\bar{y}^{(\beta)}$ minimizes the following convex optimization problem: 
\begin{equation}\label{eq: SocialCostWE}
\begin{aligned}
    \min_{y\in \R^{\numLinks}} &\quad \sum_{i\in[\numLinks]} y_i\costEdge_i(y_i) + \frac{1}{\beta}\sum_{i\in[\numLinks]} y_i\log(y_i) \\
    \text{s.t} & \quad \sum_{i\in[\numLinks]} y_i = \frac{\meanInflow}{\meanOutflow}.
\end{aligned}
\end{equation}
\end{definition}
A commonly used notion to quantify the social objective is the total latency experienced on the network \cite{roughgarden2010algorithmic,farokhi2015piecewise}.  
Note that \eqref{eq: SocialCostWE} is an \emph{entropy} regularized social objective function where the regularization weight depends on \(\beta\). As \(\beta \uparrow \infty\), the perturbed socially optimal load \(\bar{y}^{(\beta)}\) becomes the socially optimal flow, which minimizes the total latency. 
 
 The following theorem shows that the flow of the continuous-time dynamical system has a unique fixed point, and this fixed point corresponds to the perturbed socially optimal load.
 \begin{theorem}\label{theorem: fixed_point}
 The flow of the continuous-time dynamical system \eqref{eq: singularPerturbedODE} has a \emph{unique} fixed point $(\fixedPointFlow(\fixedPointPrice), \fixedPointPrice)$ such that $\fixedPointFlow(\fixedPointPrice)$ is a stochastic user equilibrium corresponding to $\fixedPointPrice$, and 
 \begin{align}\label{eq: pEquilibrium}
    \fixedPointPrice_i = \fixedPointFlow_i(\fixedPointPrice) \frac{ d \costEdge_i(\fixedPointFlow_i(\fixedPointPrice))}{d x}.
\end{align}
Moreover, $\fixedPointFlow(\fixedPointPrice)$ is the perturbed socially optimal load. 
 \end{theorem}

At $(\fixedPointFlow_i(\fixedPointPrice), \fixedPointPrice)$, the load on the network is a stochastic user equilibrium given $\fixedPointPrice$ and demand $\lambda/\mu$. This implies that travelers' routing strategy, when averaged over all time steps, is a perturbed best response given $\fixedPointPrice$ and $\lambda/\mu$. Moreover, the unique price vector $\fixedPointPrice$ is the marginal latency cost, which ensures that the stochastic user equilibrium $\fixedPointFlow(\fixedPointPrice)$ is also a perturbed socially optimal load. As $\beta \to \infty$, $\fixedPointFlow(\fixedPointPrice)$ becomes a socially optimal load.


We prove Theorem \ref{theorem: fixed_point} in three parts: Firstly, we show that for any toll vector $\priceContinuous\in \R^{\numLinks}$, the stochastic user equilibrium $\fixedPointFlow(\priceContinuous)$ exists and is unique, and can be solved as the optimal solution of a convex optimization problem (Lemma \ref{lemma:gStochasticEquilibrium}). Secondly, we show that there exists a unique toll vector $\fixedPointPrice$ that satisfies \eqref{eq: pEquilibrium} (Lemma \ref{lemma:unique_price}). This requires us to prove that the stochastic user equilibrium $\fixedPointFlow(\priceContinuous)$ is monotonic in $\priceContinuous$ (Lemma \ref{lemma:monotonicity} ). Finally, we conclude the theorem by proving that at $\fixedPointPrice$ the stochastic user equilibrium $\fixedPointFlow(\fixedPointPrice)$ is the perturbed socially optimal load (Lemma \ref{lemma:optimality}).

We now present the lemmas that are referred in each of the parts above, and provide the proof ideas of these results. We include the formal proofs in Appendix \ref{appsec: ProofLemma}. 

\emph{Part 1:} For any $p$, the load vector $\fixedPointFlow(p)$ is the unique stochastic user equilibrium. 
\begin{lemma}\label{lemma:gStochasticEquilibrium}
For every \(\priceContinuous\in \R^{\numLinks}\), \(\fixedPointFlow(\priceContinuous)\) is the unique optimal solution to the following convex optimization problem:
 \begin{equation}\label{eq: OptimizationProblemRevised}
\begin{aligned}
    \min_{y\in \R^{\numLinks}} & \quad  \sum_{i\in [\numLinks]} \int_{0}^{y_i} \costPrice_i(s,p_i)ds + \frac{1}{\beta} \sum_{i\in [\numLinks]} y_i \ln{y_i}, \\
    \text{s.t} & \quad \sum_{i\in[\numLinks]} y_i = \frac{\meanInflow}{\meanOutflow}
\end{aligned}
\end{equation}
\end{lemma}

The proof of Lemma \ref{lemma:gStochasticEquilibrium} follows by verifying that \(\fixedPointFlow(p)\) (refer Definition \ref{def:generalized_SUE}) satisfies the Karush–Kuhn–Tucker (KKT) conditions corresponding to \eqref{eq: OptimizationProblemRevised}, which is strictly convex problem and therefore has unique solution. We note that as $\beta \to \infty$, the stochastic user equilibrium $\fixedPointFlow(p)$ becomes a Wardrop equilibrium, where travelers only take routes with the minimum cost.

\emph{Part 2:} We first show that the stochastic user equilibrium $\fixedPointFlow(p)$ is monotonic in $p$. This allows us to prove the existence and uniqueness of $\fixedPointPrice$ given by \eqref{eq: pEquilibrium}.  
\begin{lemma}\label{lemma:monotonicity}(\textbf{Monotonicity of $\fixedPointFlow(\priceContinuous)$})
For any \(p,p'\in\R^{\numLinks}\) we have 
\[\lara{ \fixedPointFlow(p)-\fixedPointFlow(p'),p-p' }<0. \]
Furthermore, \(\fixedPointFlow(p)\) is a continuously differentiable function. Consequently, \(\frac{\partial \fixedPointFlow_i(\priceContinuous)}{\partial \priceContinuous_i} < 0\) for all \(i\in[\numLinks]\). 
\end{lemma}
The proof of Lemma \ref{lemma:monotonicity} is proved by using the variational inequalities of stochastic user equilibrium, which are the first-order optimality condition associated with \eqref{eq: OptimizationProblemRevised}. Moreover, the monotonicity property relies on the fact that the latency on each link increases in the load.

\begin{lemma}\label{lemma:unique_price}(\textbf{Existence and uniqueness of $\fixedPointPrice$})
The price vector $\fixedPointPrice$, defined in \eqref{eq: pEquilibrium}, exists and is unique. 
\end{lemma}
In Lemma \ref{lemma:unique_price} the  existence is proved by using the Brouwer's fixed point theorem. We prove uniqueness by contradiction. More formally, let \({p},{p}'\in \R^{\numLinks}\) be two distinct toll vectors that satisfy \eqref{eq: pEquilibrium}. Then, from \eqref{eq: pEquilibrium}, we know that for every \(i\in[\numLinks]\):
\begin{align}
    \fixedPointPrice_i - \fixedPointPrice_i' 
    &= \lr{\fixedPointFlow_i( p)-\fixedPointFlow_i(p')}   \frac{d \costEdge_i(\fixedPointFlow_i(p))}{dx}\notag\\&\quad\quad +\fixedPointFlow_i(p')\lr{\frac{ d \costEdge_i(\fixedPointFlow_i(p))}{dx}-\frac{d \costEdge_i(\fixedPointFlow_i(p'))}{dx}}.\label{eq:contradiction}
\end{align}
By multiplying  both sides of \eqref{eq:contradiction} by \((\fixedPointFlow_i(p)-\fixedPointFlow_i(p'))\),and summing over all \(i\in[\numLinks]\), we have
\begin{align*}
    \lara{ \fixedPointFlow(p)-\fixedPointFlow(p'), p-p' } \geq 0, 
\end{align*}
which contradicts the monotonicity condition proved in Lemma \ref{lemma:monotonicity}. Therefore, the price $\fixedPointPrice$ is unique. 

\emph{Part 3:} Finally, we prove that given $\fixedPointPrice$, $\fixedPointFlow(\fixedPointPrice)$ is a perturbed socially optimal load. 
\begin{lemma}\label{lemma:optimality}
The load vector $\fixedPointFlow(\fixedPointPrice)$ is the perturbed socially optimal load. 
\end{lemma}
We prove Lemma \ref{lemma:optimality} by showing that the variational inequality satisfied by the stochastic user equilibrium at \(\fixedPointPrice\) is identical to that satisfied by the perturbed socially optimal load. 

Lemmas \ref{lemma:gStochasticEquilibrium} -- \ref{lemma:optimality} conclude Theorem \ref{theorem: fixed_point}. 
\subsection{Convergence}\label{ssec: ConvergenceDiscrete}
In this section, we show that $(\flowDiscrete_i(n), \priceDiscrete_i(n))$ induced by the discrete-time stochastic update converges to a {neighborhood} of $(\fixedPointFlow(\fixedPointPrice), \fixedPointPrice)$. The size of the neighborhood depends on the load update stepsize $\mu$ and the stepsize ratio between the two dynamics $a/\mu$. That is, the tolls eventually induce a perturbed socially optimal flow.

\begin{theorem}\label{theorem:discrete_convergence}
\begin{align}\label{eq: discreteNeighborhood}
    \limsup_{n\ra+\infty}\avg[\|\flowDiscrete_n-\fixedPointFlow(\fixedPointPrice)\|^2+\|\priceDiscrete_n-\fixedPointPrice\|^2] = \mc{O}\lr{\meanOutflow+\frac{a}{\meanOutflow}}.
\end{align}
Moreover, for any \(\delta>0\):
\begin{align}\label{eq:probability_neighborhood}
    \limsup_{n\ra+\infty}\Pr(\|\flowDiscrete_n-\fixedPointFlow(\fixedPointPrice)\|^2+\|\priceDiscrete_n-\fixedPointPrice\|^2 \geq \delta) \leq \mc{O}\lr{\frac{\mu}{\delta}+\frac{a}{\delta \mu}}. 
\end{align}
\end{theorem}

In Theorem \ref{theorem:discrete_convergence}, \eqref{eq: discreteNeighborhood} provides the neighborhood of the socially optimal load and tolls, where the discrete-time stochastic updates converge to in expectation. In particular, as the step size \(\mu\) and the stepsize ratio between the two updates \(\epsilon = a/\mu\) decreases (i.e. the discrete-time step for load update is small, and the toll update is much slower than the load update), the expected value of the distance between $(\flowDiscrete(n), \priceDiscrete(n))$ and $(\fixedPointFlow(\fixedPointPrice), \fixedPointPrice)$ decreases for $n \to \infty$. Moreover, by 
applying the Markov's inequality, \eqref{eq: discreteNeighborhood} also implies that for any neighborhood of $(\fixedPointFlow(\fixedPointPrice), \fixedPointPrice)$, $(\flowDiscrete(n), \priceDiscrete(n))$ converges to that neighborhood with high probability, and this probability increases as $\mu$ and $\epsilon$ decreases.  


To prove Theorem \ref{theorem:discrete_convergence}, we need to prove the following technical lemma: 
\begin{lemma}\label{lemma: prep}
\hfill 
\begin{itemize} 
    \item[\textbf{(C1)}] For all \( i\in[\numLinks]\),  \(\{\martingale_i(n+1)\}_n\) in \eqref{eq: MartingaleDefinition} is a martingale difference sequence with respect to the filtration \(\filtration_n = \sigma(\cup_{i\in[\numLinks]} \lr{\flowDiscrete_i(1),\inflow(1),\outflow_i(1),\dots,\flowDiscrete_i(n),\inflow(n),\outflow_i(n)} )\)
    \item[\textbf{(C2)}] For all \( i\in[\numLinks]\),  \(X_i(n) \in \lr{0,\flowDiscrete_i(0)+\frac{\inflowUpper}{\outflowLower}}\). Consequently, \(\avg[\|X(n)\|^2]<+\infty\), \(\avg[\|P(n)\|^2] <+\infty\).
    \item[\textbf{(C3)}] There exists \( K >0\) such that for any $i \in [R]$ and any $n$,  \(\avg[|\martingale_i(n+1)|^2|\filtration_n] \leq K\lr{1+\|\flowDiscrete_i(n)\|^2} <+\infty.\)
    \item[\textbf{(C4)}] For any $p$, any $\tilde{x}(t): \mathbb{R}_{\geq 0} \to \mathbb{R}^{\numLinks}$ induced by following continuous-time dynamical system
    \begin{equation}
\begin{aligned}\label{eq: FixedPointFlowDynamics}
    \dot{\flowContinuousNew}_i(t) &= \vectorFieldX_i(\flowContinuousNew(t),p) -\flowContinuousNew_i(t), \quad \quad \forall i \in [\numLinks], \forall \ t\geq 0,
\end{aligned}
\end{equation}
satisfies $\lim_{t \to \infty}\tilde{x}(t)=\fixedPointFlow(p)$. Futhermore, \(\fixedPointFlow(p)\) is Lipschitz.
\item[\textbf{(C5)}] Any $\tilde{p}(t): \mathbb{R}_{\geq 0} \to \mathbb{R}^\numLinks$ induced by following continuous-time dynamical system
\begin{align}\label{eq: PriceDynamicsTilde}
    \dot{\priceContinuousNew}_i(t) = -\priceContinuousNew_i(t)+\fixedPointFlow_i(\priceContinuousNew(t))\frac{d  \costEdge_i(\fixedPointFlow_i(\priceContinuousNew(t)))}{d x}, \quad t\geq 0,
\end{align}
satisfies $\lim_{t \to \infty} \priceContinuousNew(t) = \fixedPointPrice$.
\end{itemize}
\end{lemma}

In Lemma \ref{lemma: prep}, condition \textbf{(C1)} relies on the fact that both the incoming and outgoing loads are i.i.d.. Condition \textbf{(C2)} ensures the boundedness of the loads and the tolls in the discrete-time stochastic updates, and it relies on the boundedness of $\inflow_n$ and $\outflow_n$. Condition \textbf{(C3)} ensures the boundedness of the martingale sequence $\{M_i(n+1)\}_n$, and it is built on condition \textbf{(C1)} and \textbf{(C2)}. 

In conditions \textbf{(C4)}, the continuous-time dynamical system \eqref{eq: FixedPointFlowDynamics} is associated with the discrete-time load update \eqref{eq: FlowDynamics} when the toll is set as a constant $p$. That is, due to the fact that toll update is at a slower timescale compared with the load update ($\epsilon \ll 1$), the continuous-time dynamical system of load evolves as if the toll is a constant $p$. We prove condition \textbf{(C4)} by showing that \eqref{eq: FixedPointFlowDynamics} is cooperative (see Theorem \ref{thm: HirshResult} in the appendix). Condition \textbf{(C4)} ensures that the load of the continuous-time dynamical system converges to the stochastic user equilibrium (Definition \ref{def:generalized_SUE}) with respect to \(p\). Recall from Lemma \ref{lemma:gStochasticEquilibrium}, the stochastic user equilibrium is unique.

On the other hand, in condition \textbf{(C5)}, \eqref{eq: PriceDynamicsTilde} is associated with \eqref{eq: PricingDynamics} when the load -- which is updated at a faster timescale -- has already converged to the stochastic user equilibrium with respect to $p(t)$. Similar to condition \textbf{(C4)}, we show that \eqref{eq: PriceDynamicsTilde} is a cooperative dynamical system. The proof of this condition is built on the monotonicity of stochastic user equilibrium with respect to the toll (Lemma \ref{lemma:monotonicity}) and the uniqueness of $\fixedPointPrice$ (Lemma \ref{lemma:unique_price}). The proofs of \textbf{(C1)} -- \textbf{(C5)} are in Appendix \ref{appsec: ProofLemma}.


Based on Lemma \ref{lemma: prep}, we can apply the theory of two timescale stochastic approximation with constant stepsizes: 
\begin{lemma}[\cite{borkar2009stochastic}]\label{lemma: borkar}
Given $\epsilon \ll 1$, \eqref{eq: discreteNeighborhood} is satisfied under \textbf{(C1)} -- \textbf{(C5)}. 
\end{lemma}

Lemmas \ref{lemma: prep} and \ref{lemma: borkar} proves \eqref{eq: discreteNeighborhood}. Additionally, \eqref{eq:probability_neighborhood} can be directly derived from \eqref{eq: discreteNeighborhood} using the Markov's inequality. Thus, we can conclude Theorem \ref{theorem:discrete_convergence}.

\section{Numerical Experiments}\label{sec: NumExp}

In this section we present numerical experiments for the results presented in \S\ref{sec: Results}. 
We observe that loads and tolls concentrate on a neighborhood of the socially optimal loads and tolls, and the continuous-time dynamical system \eqref{eq: singularPerturbedODE} closely approximates the discrete-time stochastic updates \eqref{eq: FlowDynamics} -- \eqref{eq: PricingDynamics}. Our numerical results are consistent with Theorems \ref{theorem: fixed_point}  and \ref{theorem:discrete_convergence}. 



Consider a network with six parallel links (i.e. $R=6$) with the following link latency functions: 
\begin{align}
    \costEdge_i(x) = i x^2 + i, \quad \quad \forall \ i \in [\numLinks].
\end{align}

We set the total time steps of discrete-time stochastic update as $N=2000$, and the dispersion parameter $\beta = 100$. We conduct the four sets of experiments with the following parameters: 
    
    (\textbf{S1}) $\meanInflow= 0.1$, $\meanOutflow = 0.05$ and \(a=0.0015\);
    
    (\textbf{S2}) $\meanInflow= 0.2$, $\meanOutflow = 0.05$ and \(a=0.0015\);
    
    (\textbf{S3})  $\meanInflow= 0.1$, $\meanOutflow = 0.05$ and \(a=0\);
    
    (\textbf{S4}) $\meanInflow= 0.2$, $\meanOutflow = 0.05$ and \(a=0\).

We note that tolls are updated with positive stepsizes in (\textbf{S1}) -- (\textbf{S2}), but remain zero in (\textbf{S3}) -- (\textbf{S4}). Also, the mean incoming load of travelers in each step $\meanInflow$ is high in (\textbf{S1}) and (\textbf{S3}), and low in (\textbf{S2}) -- (\textbf{S4}). In Fig. \ref{fig:6LinkSetting1}, we demonstrate the loads and tolls obtained in the discrete-time stochastic updates \eqref{eq: FlowDynamics} -- \eqref{eq: PricingDynamics} (represented by dots), and those induced by the continuous-time dynamical system \eqref{eq: singularPerturbedODE} (represented by solid lines) in (\textbf{S1}) and (\textbf{S2}), respectively. In Fig. \ref{fig:6LinkNoPrice}, we demonstrate the loads in the discrete-time updates and the socially optimal load for (\textbf{S3}) and (\textbf{S4}). We omit the figures for tolls since tolls are not updated with step size $a=0$.

We now present the main observations from the numerical study.
Firstly, we observe that in Fig. \ref{fig:6LinkSetting1}, the trajectories of continuous-time dynamical system \eqref{eq: singularPerturbedODE} closely track the discrete-time stochastic load update. Moreover, we observe that it takes more time steps for the tolls (refer Fig.\ref{fig:6LinkPriceProfileL1}, Fig.\ref{fig:6LinkPriceProfileL2}) to converge compared to the loads. This is because tolls are updated at a slower timescale (i.e. \(a\ll\mu\)). 

Secondly, we show that the repeatedly updated tolls eventually induce the perturbed socially optimal load in Fig. \ref{fig:6LinkFlowProfileL1} and \ref{fig:6LinkFlowProfileL2}. On the other hand, when tolls are zero (i.e. inactive) in all steps, loads converge to a stochastic user equilibrium (Fig. \ref{fig:6LinkPriceProfileL1NoToll} and \ref{fig:6LinkPriceProfileL2NoToll}), which is different from the socially optimal load. This means that links are inefficiently utilized when tolls are inactive. We observe that the low cost links (links 1, 2, and 3) are disproportionately over-utilized while the remaining links are underused especially in the setting with high incoming load $\lambda$ (Fig.\ref{fig:6LinkPriceProfileL2NoToll}).  

 Thirdly, we note that both the perturbed socially optimal load and the stochastic user equilibrium change with the average incoming load $\lambda$. In particular, more links are utilized at fixed point with high $\lambda$ (Fig. \ref{fig:6LinkFlowProfileL2} for \textbf{(S2)} and Fig. \ref{fig:6LinkPriceProfileL2NoToll} for \textbf{(S4)}) compared to that with low $\lambda$ (Fig. \ref{fig:6LinkFlowProfileL1} for \textbf{(S2)} and Fig. \ref{fig:6LinkPriceProfileL1NoToll} for \textbf{(S4)}).

\begin{figure}[htp]
\begin{subfigure}{.4\textwidth}
  \centering
  \includegraphics[scale = 0.4]{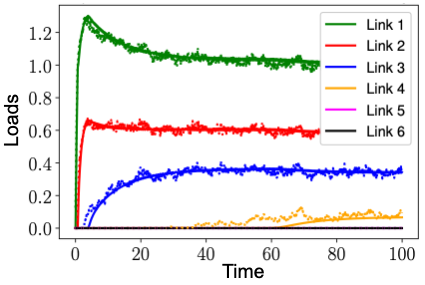}  
  \caption{Loads in (\textbf{S1})}
  \label{fig:6LinkFlowProfileL1}
\end{subfigure}
\begin{subfigure}{.4\textwidth}
  \centering
  \includegraphics[scale = 0.4]{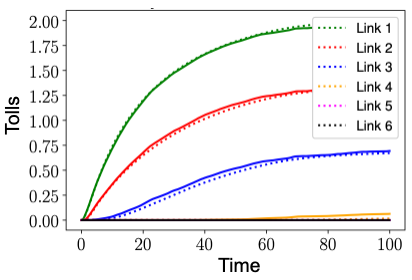}
  \caption{Tolls in (\textbf{S1})}
  \label{fig:6LinkPriceProfileL1}
\end{subfigure}
\begin{subfigure}{.4\textwidth}
  \centering
  \includegraphics[scale = 0.4]{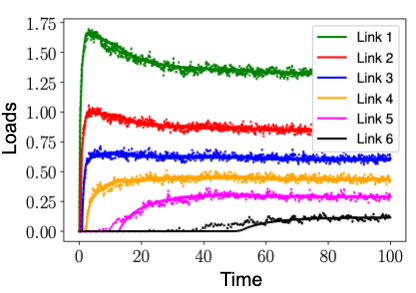}  
  \caption{Loads in (\textbf{S2})}
  \label{fig:6LinkFlowProfileL2}
\end{subfigure}
\begin{subfigure}{.4\textwidth}
  \centering
  \includegraphics[scale = 0.4]{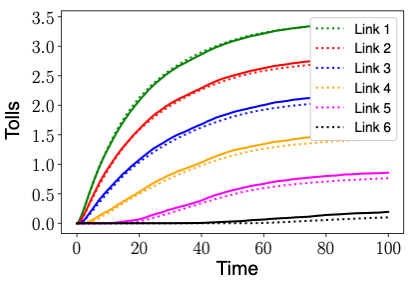}  
  \caption{Tolls in (\textbf{S2})}
  \label{fig:6LinkPriceProfileL2}
\end{subfigure}
\caption{Loads and Tolls in discrete-time stochastic update \eqref{eq: FlowDynamics} -- \eqref{eq: PricingDynamics} (dots) and continuous-time dynamical system \eqref{eq: singularPerturbedODE} (solid lines) in \textbf{(S1)} and \textbf{(S2)}.}
\label{fig:6LinkSetting1}
\end{figure}

\begin{figure}[htp]
\begin{subfigure}{.4\textwidth}
  \centering
  \includegraphics[scale = 0.4]{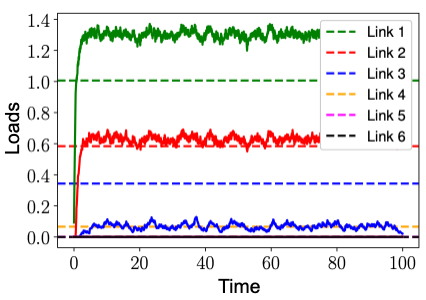}  
  \caption{ Loads in (\textbf{S3})}
  \label{fig:6LinkPriceProfileL1NoToll}
\end{subfigure}
\begin{subfigure}{.4\textwidth}
  \centering
  \includegraphics[scale = 0.4]{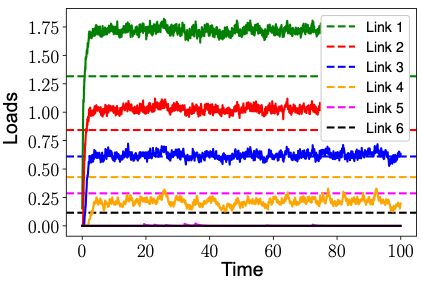}
  \caption{Loads in (\textbf{S4})}
  \label{fig:6LinkPriceProfileL2NoToll}
\end{subfigure}
\caption{Comparison between loads induced by the discrete-time stochastic update \eqref{eq: FlowDynamics}(solid lines) when \(P_n=0\) and the socially optimal load (dashed lines) in \textbf{(S3)} and \textbf{(S4)}.} 
    \label{fig:6LinkNoPrice}
\end{figure}

\pagebreak 
\pagebreak

\section{Conclusions and Future Work}\label{sec: Conclusion}
We propose a two timescale discrete-time stochastic dynamics that captures the joint evolution of loads in parallel network and adaptive adjustment of tolls. We analyze the properties of fixed points and the convergence of loads and tolls in this dynamics. We find that under our dynamics, the tolls eventually induce a socially optimal load with high probability. Our results allow a central authority to set tolls in networks \textit{adaptively} and \emph{optimally} in response to the dynamic arrival of travelers who myopically make routing decisions. 

One future direction of interest is to extend our results on parallel networks to general networks. The challenge of setting tolls on general networks is to account for the fact that travelers may change their route choices at intermediate nodes of the network due to the change of loads on links during their trips. Therefore, the toll updates need to account for the myopic route choice of incoming travelers as well as the change of routing decisions for travelers in the network. 

Another extension of this work is to design toll pricing that only set tolls on a subset of links instead of the entire network. This extension is of practical importance since setting high tolls on many links is difficult in implementation, and is not socially desirable. 

\appendix
\section{Convergence for cooperative dynamical systems}\label{appsec: Cooperative}
In this section we review a result from \cite{hirsch1985systems} which provides an easily verifiable requirement on the vector field which ensures convergence to equilibrium.  

A dynamical system 
\begin{align}\label{eq: DynSys}\tag{A.1}
\dot{x}(t) = f(x(t))
\end{align}
with a \(C^1-\) function \(f:\R^n\rightarrow \R^n\) is \emph{cooperative} if:

\begin{enumerate}[label=(P-\roman*), align=left,leftmargin=*,widest=iiii]
    \item \label{enum: Assm1} for any \(i\neq j \in [n]\) and \(x\in \R^n\) we have \(\frac{\partial f_i(x)}{\partial x_j} \geq 0\); 
    \item \label{enum: Assm2} the Jacobian matrix \(\nabla f(x)\) is irreducible;\footnote{A matrix $A = [A_{ij}] \in \mathbb{R}^{n \times n}$ is irreducible if whenever the set $\{1, \dots, n\}$ is expressed as the union of two disjoint proper subsets $S, S'$, then for every $i \in S$ there exist $j,k \in S'$ such that $A_{ij} \neq 0$ and $A_{ki} \neq 0$.}
    \item \label{enum: Assm3} for every \(i\in[n]\), \(f_i(0)\geq 0\);
    \item \label{enum: Assm4} for any \(x\in\Rp^n\) there exists \(y > x\) with \(f_i(y) < 0\) for all \(i\in[n]\)
 \end{enumerate}
 \begin{theorem}[{\cite[Theorem 5.1]{hirsch1985systems}}]\label{thm: HirshResult}
 If the dynamical system \eqref{eq: DynSys} is cooperative (i.e. satisfies the assumptions \ref{enum: Assm1}-\ref{enum: Assm4}), trajectories starting almost everywhere in \(\Rp^n\) converge to the set \(\{x: f(x) = 0\}\). 
 \end{theorem}
\section{Proofs}\label{appsec: ProofLemma}
\subsection{{Proof of Lemma \ref{lemma:gStochasticEquilibrium}}}
In the subsequent proof we shall show that the following claims hold: 

\begin{enumerate}[label= (C-\Roman*),widest=ii,leftmargin=*,align=left]
    \item \label{thm: claim1} For any \(p\in\R^{\numLinks}\), the optimizer of \eqref{eq: OptimizationProblemRevised} is unique; call it \(y^\star(p)\);
    \item \label{thm: claim2} \(y^\star(p)\) is an optimal solution if and only if it satisfies \(h(y^\star(p),p) = y^\star\)
\end{enumerate}
Note that the above two claims ensure that \(\fixedPointFlow(p) = y^\star(p)\). The subsequent exposition establishes the validity of these claims.

To see that \ref{thm: claim1}  is true, we note that the feasible set of \eqref{eq: OptimizationProblemRevised} is a compact convex set.  Moreover, from the convexity of $\costEdge$ it follows that \(
    \sum_{i\in [\numLinks]} \int_{0}^{y_i} \costPrice_i(s,p_i)ds + \frac{1}{\beta} \sum_{i\in [\numLinks]} y_i \ln{y_i}
\) is a strictly convex function in \(y\). Strict convexity ensures the uniqueness of solution of a convex optimization problem on a convex compact set \cite{boyd2004convex}.

For \ref{thm: claim2}, we now employ KKT conditions of optimality: let \(\delta\in\R \) be the Lagrange multiplier corresponding to the equality constraint. Define the  Lagrangian as follows:
\[
L(y,\delta,p) \defas \sum_{i\in [\numLinks]} \int_{0}^{y_i} \costPrice_i(s,p_i)ds + \frac{1}{\beta} \sum_{i\in \numLinks} y_i \ln{y_i} + \delta \lr{\sum_{i\in[\numLinks]} y_i - \frac{\meanInflow}{\meanOutflow}}
\]

The optimal solution $y^{\star}(p)$ and Lagrange multiplier $\delta^\star$ satisfy: \begin{enumerate}[widest=ii,leftmargin=*,align=left]
    \item \(\nabla_{y}L(y^{\star}(p),\delta^{\star},p) = 0\). This gives us \(\costPrice_i(y_i^{\star}(p),p_i) + \frac{1}{\beta}\lr{1+\ln(y_i^{\star}(p))} +\delta^{\star} = 0\), which can be also written as follows
    \begin{align}
    \label{eq: Opt_y}
    y_i^{\star}(p) = \exp(-\beta \delta^\star-1)\exp(-\beta \costPrice_i(y_i^{\star}(p),p_i)), \quad \forall \ i\in[\numLinks].
    \end{align}
    \item \(\nabla_\delta L(y^{\star}(p),\delta^{\star},p) = 0\). That is, 
    \(
    \sum_{i\in \numLinks} y_i^{\star}(p)  = \frac{\meanInflow}{\meanOutflow}.
    \)
    \end{enumerate}
    From \eqref{eq: Opt_y}, summing over \(i\in[\numLinks]\) on both the sides, we obtain 
    \[
    \frac{\meanInflow}{\meanOutflow} = \exp(-\beta \delta^\star -1 ) {\sum_{i\in[\numLinks]}\exp(-\beta \costPrice_i(y_i^{\star}(p),p_i))}
    \]

The above conditions give us that \(y^\star(p)\) is an optimal solution of \eqref{eq: OptimizationProblemRevised} if and only if it satisfies
\[
y_i^\star(p) = \frac{\meanInflow}{\meanOutflow} \frac{\exp(-\beta \costPrice_i(y_i^\star(p),p_i))}{\sum_{i\in[\numLinks]}\exp(-\beta \costPrice_i(y_i^\star(p),p_i))} = \vectorFieldX_i(y^\star(p),\priceContinuous), \quad \forall \ i\in[\numLinks].
\]

\subsection{Proof of Lemma \ref{lemma:monotonicity}}
Before proving Lemma \ref{lemma:monotonicity}, we will first state the following result that will be useful later. 

 Define \(\diagC(x) \defas \diag\lr{\lr{\frac{d \costEdge_i(x)}{dx}}_{i=1}^{\numLinks}} \) and \(\diagMarginal(x) \defas \diag\lr{\lr{x_i\frac{d \costEdge_i(x)}{dx}}_{i=1}^{\numLinks}}\).
\begin{lemma}\label{lem: LipschitzFunction}
The function $h(x,p)$ presented in \eqref{eq: hVectorFieldDef} is Lipschitz in \(x\) and \(p\) and satisfies:
\begin{align}
    \grad_x \vectorFieldX(\fixedPointFlow(p),p) = \frac{\meanInflow}{\meanOutflow}\beta\fixedPointFlow(p)\fixedPointFlow(p)^\top \diagC(\fixedPointFlow(p)) - \beta\diagMarginal\lr{\fixedPointFlow(p)} 
\end{align}
\end{lemma}
\begin{proof}[Proof of Lemma \ref{lem: LipschitzFunction}]
To show that the function is Lipschitz, it is sufficient to show that the norm of the gradient of this function is bounded. This is due to the first order Taylor expansion of the function. 

Therefore, for every \(i\in[\numLinks]\) we first compute \(\grad_x h_i(x,p)\) entrywise as follows: 
\begin{equation}\label{eq: xDerivative}
\begin{aligned}
\frac{\partial h_i(x,p)}{\partial x_j} = \begin{cases}  \frac{\meanInflow}{\meanOutflow} \frac{-\beta\frac{d\costEdge_i(x_i)}{dx}\sum_{k\in[R]}\exp(-\sueParameter \costPrice_k(x_k,p_k))  \exp(-\beta \costPrice_i(x_i,p_i)) +\beta\frac{d\costEdge_i(x_i)}{dx} \lr{\exp(-\beta \costPrice_i(x_i,p_i))}^2 }{\lr{\sum_{k\in[R]}\exp(-\sueParameter \costPrice_k(x_k,p_k)) }^2} & \text{if}\ \ i=j;\\
\frac{\meanInflow}{\meanOutflow}\beta\frac{d\costEdge_j(x_j)}{dx}\exp(-\beta \costPrice_i(x_i,p_i))\frac{\exp(-\beta \costPrice_j(x_j,p_j))}{\lr{\sum_{k\in[R]}\exp(-\sueParameter \costPrice_k(x_k,p_k)) }^2} & \text{otherwise}
\end{cases}
\end{aligned}
\end{equation}
Evaluating the above derivative at \(\fixedPointFlow(p)\)
\begin{equation}\label{eq: xDerivativeFP}
\begin{aligned}
\frac{\partial h_i(\fixedPointFlow(p),p)}{\partial x_j} = \begin{cases}  \beta \frac{d \costEdge_i(\fixedPointFlow_i(p))}{dx} \lr{ -\fixedPointFlow_i(p) +\frac{\meanOutflow}{\meanInflow}\lr{\fixedPointFlow_i(p)}^2 } & \text{if}\ \ i=j;\\
\frac{\meanOutflow}{\meanInflow}\beta\frac{d\costEdge_j(\fixedPointFlow_j(p))}{dx}\fixedPointFlow_i(p)\fixedPointFlow_j(p) & \text{otherwise}
\end{cases}
\end{aligned}
\end{equation}
To state it concisely:
\begin{align}\label{eq: GradHEq}
\grad_x \vectorFieldX(\fixedPointFlow(p),p) = \frac{\meanInflow}{\meanOutflow}\beta\fixedPointFlow(p)\fixedPointFlow(p)^\top \diagC(\fixedPointFlow(p)) - \beta\diagMarginal\lr{\fixedPointFlow(p)} 
\end{align}
Similarly, 
\begin{equation}\label{eq: pDerivative}
\begin{aligned}
\frac{\partial h_i(x,p)}{\partial p_j} = \begin{cases}  \frac{\meanInflow}{\meanOutflow} \frac{-\beta\sum_{k\in[R]}\exp(-\sueParameter \costPrice_k(x_k,p_k))  \exp(-\beta \costPrice_i(x_i,p_i)) +\beta \lr{\exp(-\beta \costPrice_i(x_i,p_i))}^2 }{\lr{\sum_{k\in[R]}\exp(-\sueParameter \costPrice_k(x_k,p_k)) }^2} & \text{if}\ \ i=j;\\
\frac{\meanInflow}{\meanOutflow}\beta\exp(-\beta \costPrice_i(x_i,p_i))\frac{\exp(-\beta \costPrice_j(x_j,p_j))}{\lr{\sum_{k\in[R]}\exp(-\sueParameter \costPrice_k(x_k,p_k)) }^2} & \text{otherwise}
\end{cases}
\end{aligned}
\end{equation}
To conclude the proof we observe that on a bounded domain the derivative in \eqref{eq: xDerivative} and \eqref{eq: pDerivative} are bounded. 
\end{proof}

\subsubsection{\textbf{Proof of Lemma \ref{lemma:monotonicity}}}
Recall from the proof of Lemma \ref{lemma:gStochasticEquilibrium} (\ref{thm: claim2} to be precise) that for any \(p\in\R^{\numLinks}\), \(\fixedPointFlow(p)\) is a solution to the optimization problem \eqref{eq: OptimizationProblemRevised}. Let the feasible set in the optimization problem \eqref{eq: OptimizationProblemRevised} be denoted by \(F\). Using first order conditions for constrained optimization we obtain the variational inequality: 
\begin{align}\label{eq: ConstrainedOptimalityP}
    \sum_{i\in[\numLinks]} \lr{\costPrice_i\lr{\fixedPointFlow_i(p),p} + \frac{1}{\beta}\lr{\ln\lr{\fixedPointFlow_i(p)}+1}}\lr{x_i-\fixedPointFlow_i(p)} \geq0 \quad \forall x \in F.
\end{align}
Similarly writing the above condition for \(p'\), we obtain \begin{align}\label{eq: ConstrainedOptimalityPD}
    \sum_{i\in[\numLinks]} \lr{\costPrice_i\lr{\fixedPointFlow_i(p'),p'} + \frac{1}{\beta}\lr{\ln\lr{\fixedPointFlow_i(p')}+1}}\lr{y_i-\fixedPointFlow_i(p')} \geq0 \quad \forall y \in F.
\end{align}
Choosing \(x=\fixedPointFlow(p')\)  in \eqref{eq: ConstrainedOptimalityP} and \(y=\fixedPointFlow(p)\) in \eqref{eq: ConstrainedOptimalityPD} and adding the two equations we obtain
\begin{align*}
    &\sum_{i\in[\numLinks]} \lr{ \costPrice_i(\fixedPointFlow_i(p),p)-\costPrice_i(\fixedPointFlow_i(p'),p')}\lr{\fixedPointFlow_i(p')-\fixedPointFlow_i(p)} \\ &\quad \quad +\sum_{i\in[\numLinks]} \lr{ \frac{1}{\beta} \ln(\fixedPointFlow_i(p))-\ln(\fixedPointFlow_i(p')) }\lr{\fixedPointFlow_i(p')-\fixedPointFlow_i(p)} \geq 0.
\end{align*}
From the fact that prices enter additively, we have 
\begin{align*}
    &\sum_{i\in[\numLinks]} \lr{ \costEdge_i(\fixedPointFlow_i(p))-\costEdge_i(\fixedPointFlow_i(p'))}\lr{\fixedPointFlow_i(p')-\fixedPointFlow_i(p)}\\ &\quad\quad + \sum_{i\in[\numLinks]} \lr{ p_i-p_i'+ \frac{1}{\beta} \ln(\fixedPointFlow_i(p))-\ln(\fixedPointFlow_i(p')) }\lr{\fixedPointFlow_i(p')-\fixedPointFlow_i(p)} \geq 0.
\end{align*}
This gives 
\begin{align*}
    &\lara{\fixedPointFlow(p)-\fixedPointFlow(p'),p-p'} \\&\leq \lara{\costEdge(\fixedPointFlow(p))-\costEdge(\fixedPointFlow(p')),\fixedPointFlow(p')-\fixedPointFlow(p)}\\&\quad +\frac{1}{\beta} \lara{ \ln(\fixedPointFlow(p)) - \ln(\fixedPointFlow(p')), \fixedPointFlow(p')-\fixedPointFlow(p)}\\
&< 0,
\end{align*}
where the last inequality follows due to the monotonicity of the cost function and natural logarithm.

 Next, we show that $\fixedPointFlow(\priceContinuous)$ is a differentiable function of $\priceContinuous$.  Recall the notation \(\diagC(x) \defas \diag\lr{\lr{\frac{d \costEdge_i(x)}{dx}}_{i=1}^{\numLinks}} \) and \(\diagMarginal(x) \defas \diag\lr{\lr{x_i\frac{d \costEdge_i(x)}{dx}}_{i=1}^{\numLinks}}\).

 To show the differentiability of \(\fixedPointFlow(p)\) with respect to \(p\), we invoke the implicit function theorem for the map $\fixedGap(x, p) = x - h(x, p)$.
Note that for any fixed \(p\), \(\fixedPointFlow(p)\) is the zero of \(\fixedGap(\cdot,p)\). To satisfy the hypothesis for the implicit function theorem, for any \(\priceContinuous\) we compute the Jacobian of the function $\fixedGap(x, p)$ with respect to \(x\) and evaluate it at $\fixedPointFlow(\priceContinuous)$. Using Lemma \ref{lem: LipschitzFunction} the Jacobian is given by: 
\begin{align*}
   \grad_x \fixedGap(\fixedPointFlow(p),p) &= I - \grad_x \vectorFieldX(\fixedPointFlow(p),p) \\
   & = I - \frac{\meanOutflow}{\meanInflow}\beta\fixedPointFlow(p)\fixedPointFlow(p)^\top \diagC(\fixedPointFlow(p)) + \beta\diagMarginal\lr{\fixedPointFlow(p)} \\ 
   & = \lr{\lr{\diagMarginal^{-1}(\fixedPointFlow(p)) + \beta I}-\beta v\mathbbm{1}^\top}\diagMarginal(\fixedPointFlow(p)) 
\end{align*}
where \(v \defas \frac{\meanOutflow}{\meanInflow}\fixedPointFlow(p)\). Note that from \eqref{eq: xEquilibrium} we have \(\mathbbm{1}^\top v = 1\).
We claim that \(\det(\grad_x \fixedGap(\fixedPointFlow(p),p)) \neq 0\). Indeed, 
\begin{align*}
    \det(\grad_x \fixedGap(\fixedPointFlow(p),p)) &= \det\lr{\diagMarginal(\fixedPointFlow(p)} \det\lr{\lr{\diagMarginal^{-1}(\fixedPointFlow(p)) + \beta I}-\beta v\mathbbm{1}^\top }\\ 
   &= \det\lr{\diagMarginal(\fixedPointFlow(p)}\cdot
    \det\lr{I-\beta v\mathbbm{1}^\top\lr{\diagMarginal^{-1}(\fixedPointFlow(p)) + \beta I}^{-1} }\\&\quad\quad \cdot \det\lr{\diagMarginal^{-1}(\fixedPointFlow(p)) + \beta I}^{-1}
\end{align*}
It is sufficient to show that \( \det\lr{I-\beta v\mathbbm{1}^\top\lr{\diagMarginal^{-1}(\fixedPointFlow(p)) + \beta I}^{-1} }\neq 0\). Using the Sherman-Morrison formula \cite{gentle2007matrix} it is necessary and sufficient to show that \(\beta\mathbbm{1}^\top\lr{\diagMarginal^{-1}(\fixedPointFlow(p)) + \beta I}^{-1} v \neq 1 \). 

We show it by contradiction. Suppose \(\beta\mathbbm{1}^\top\lr{\diagMarginal^{-1}(\fixedPointFlow(p)) + \beta I}^{-1} v = 1 \). Then 
\begin{align*}
&\implies \quad \sum_{i=1}^{\numLinks} v_i \frac{\beta}{\frac{1}{m_i}+\beta} = 1\\
&\implies \quad \sum_{i=1}^{\numLinks} v_i \frac{\beta}{\frac{1}{m_i}+\beta} = \sum_{i=1}^{\numLinks} v_i \\ 
&\implies \sum_{i=1}^{\numLinks} v_i\frac{\frac{1}{m_i}}{\beta+\frac{1}{m_i}} = 0
\end{align*}
This is not possible as the terms on the left hand side are all positive.

The claim that \(\frac{\partial \fixedPointFlow_i(\priceContinuous)}{\partial \priceContinuous_i} < 0\) also follows from the monotonicity of \(\fixedPointFlow(\priceContinuous)\). Indeed, if we choose two price profiles that differ only at one index then the monotonicity property of \(\fixedPointFlow(p)\) ensures \(\frac{\partial \fixedPointFlow_i(\priceContinuous)}{\partial \priceContinuous_i} < 0\). 
This completes the proof.
\subsection{Proof of Lemma \ref{lemma:unique_price}}
We shall first show the existence and then prove the uniqueness of \(\fixedPointPrice\). 

For any \(p\in\R^{\numLinks}, i\in[\numLinks]\), define \(z_i(p) = \fixedPointFlow(p)\frac{d \costEdge_i(\fixedPointFlow_i(p))}{dx}\). The equilibrium price \(\fixedPointPrice\) is then solution to the fixed point equation \(p=z(p)\) where $z(p) = [z_i(p)]_{i\in [\numLinks]}$.  Define a set \(K = \{y\in\R^{\numLinks}: y\geq 0, \|y\|_1 \leq \frac{\meanInflow}{\meanOutflow}\max_{i\in[\numLinks]}\frac{d c_i\lr{\frac{\meanInflow}{\meanOutflow}}}{dx}\}\). Observe that \(z(\cdot)\) maps the convex compact set \(K\) to itself. Therefore, Brouwer's fixed point theorem \cite{sastry2013nonlinear} guarantees the existence of fixed point \(\fixedPointPrice\).

Now, we prove the uniqueness of the fixed point $\fixedPointPrice$ satisfying \eqref{eq: pEquilibrium}.  
We shall prove this via a contradiction argument. Assume that there are two distinct price profiles $\fixedPointPrice, \fixedPointPrice' \in \R^{\numLinks}$ that both satisfy the fixed point equation \eqref{eq: pEquilibrium}, therefore for every \(i\in[\numLinks]\):
\begin{align*}
    \fixedPointPrice_i &= \fixedPointFlow_i(\fixedPointPrice) \frac{d \costEdge_i(\fixedPointFlow_i(\fixedPointPrice))}{dx} \\ \fixedPointPrice_i' &= \fixedPointFlow_i(\fixedPointPrice') \frac{d \costEdge_i(\fixedPointFlow_i(\fixedPointPrice'))}{dx}
\end{align*}
Taking the difference we get 
\begin{align*}
    \fixedPointPrice_i - \fixedPointPrice_i' &= \fixedPointFlow_i(\fixedPointPrice) \frac{d \costEdge_i(\fixedPointFlow_i(\fixedPointPrice))}{dx} - \fixedPointFlow_i(\fixedPointPrice') \frac{d \costEdge_i(\fixedPointFlow_i(\fixedPointPrice'))}{dx} , \\ 
    &= \fixedPointFlow_i(\fixedPointPrice) \frac{d \costEdge_i(\fixedPointFlow_i(\fixedPointPrice))}{dx} - \fixedPointFlow_i(\fixedPointPrice') \frac{d \costEdge_i(\fixedPointFlow_i(\fixedPointPrice))}{dx}\\
    &\quad \quad \quad + \fixedPointFlow_i(\fixedPointPrice') \frac{d \costEdge_i(\fixedPointFlow_i(\fixedPointPrice))}{dx} - \fixedPointFlow_i(\fixedPointPrice') \frac{d \costEdge_i(\fixedPointFlow_i(\fixedPointPrice'))}{dx},  \\
    &= \lr{\fixedPointFlow_i( \fixedPointPrice)-\fixedPointFlow_i(\fixedPointPrice')}   \frac{d\costEdge_i(\fixedPointFlow_i(\fixedPointPrice))}{dx}\\&\quad\quad\quad +\fixedPointFlow_i(\fixedPointPrice')\lr{\frac{d \costEdge_i(\fixedPointFlow_i(\fixedPointPrice))}{dx}-\frac{d \costEdge_i(\fixedPointFlow_i(\fixedPointPrice'))}{dx}},
\end{align*}
for every \(i\in[\numLinks]\). Multiplying \((\fixedPointFlow_i(\fixedPointPrice)-\fixedPointFlow_i(\fixedPointPrice'))\) in the preceding equation we obtain
\begin{equation}
\begin{aligned}\label{eq: pUniquenessContradiction}
    &(\fixedPointFlow_i(\fixedPointPrice)-\fixedPointFlow_i(\fixedPointPrice'))(\fixedPointPrice_i - \fixedPointPrice_i') = \lr{\fixedPointFlow_i(\fixedPointPrice)-\fixedPointFlow_i(\fixedPointPrice')}^2\frac{d \costEdge_i(\fixedPointFlow_i(\fixedPointPrice))}{dx}\\&\quad\quad\quad +\fixedPointFlow_i(\fixedPointPrice') \lr{\frac{d \costEdge_i(\fixedPointFlow_i(\fixedPointPrice))}{dx}-\frac{d \costEdge_i(\fixedPointFlow_i(\fixedPointPrice'))}{dx}} (\fixedPointFlow_i(\fixedPointPrice)-\fixedPointFlow_i(\fixedPointPrice')), 
\end{aligned}
\end{equation}
for every \(i\in[\numLinks]\). Convexity of the edge cost function ensures that right hand side of \eqref{eq: pUniquenessContradiction} is always non-negative for every \(i\in[\numLinks]\). 
Summing over \(i\in[R]\) and using Lemma \ref{lemma:monotonicity} we obtain:
\begin{align*}
    0>\lara{ \fixedPointFlow_i(\fixedPointPrice)-\fixedPointFlow_i(\fixedPointPrice'), \fixedPointPrice-\fixedPointPrice' } \geq 0.
\end{align*}
which contradicts our original hypothesis that there are two distinct price profiles satisfying \eqref{eq: pEquilibrium}. 

\subsection{Proof of Lemma \ref{lemma:optimality}}
Note that for any \(i\in[\numLinks]\)
\begin{align}\label{eq: PriceBeta}
    \fixedPointPrice_i = \fixedPointFlow_i(\fixedPointPrice) \frac{d \costEdge_i(\fixedPointFlow_i(\fixedPointPrice))}{dx}
\end{align}

Let 
\begin{align*}
   x^\star = \argmin_{y\in F} \lr{ \sum_{i\in[\numLinks]} y_i\costEdge_i(y_i) + \frac{1}{\beta}\sum_{i\in[\numLinks]} y_i\log(y_i)},
\end{align*}
where \(F= \{y: \sum_{i\in[\numLinks]}y_i = \frac{\meanInflow}{\meanOutflow}\}\) is the feasible set. We claim that \(x^\star =  \fixedPointFlow(\fixedPointPrice)\).

To show this claim, we first note that the above problem is a strictly convex optimization problem. The necessary and sufficient conditions for constrained optimality ensures that for any \(y\in F\): 
\begin{align}\label{eq: SocialCostFirst}
    \sum_{i\in[\numLinks]} \lr{ \costEdge_i(x^\star_i) + x_i^\star\frac{d\costEdge_i(x_i^\star)}{dx} + \frac{1}{\beta} \lr{1+\log(x_i^\star)}}\lr{ y_i-x_i^\star} \geq 0. 
\end{align}

Now recall that \(\fixedPointFlow(\fixedPointPrice)\) satisfies 
\begin{align*}
    \fixedPointFlow(\fixedPointPrice) = \argmin_{y\in F} \sum_{i\in[\numLinks]} \int_{0}^{y_i}\lr{\costEdge_i(z)+\fixedPointPrice_i } \textsf{d}z+\frac{1}{\beta}\sum_{i\in[\numLinks]} y_i\log(y_i).
\end{align*}
Note that we have already noted in proof of Lemma \ref{lemma:gStochasticEquilibrium} that the above optimization problem is strictly convex. The constrained optimality conditions ensure that for any \(y\in F\):
\begin{align}\label{eq: SelfishCostFirst}
    \sum_{i\in[\numLinks]}\lr{\costEdge_i(\fixedPointFlow_i(\fixedPointPrice)) + \fixedPointPrice_i + \frac{1}{\beta} \lr{1+\log(\fixedPointFlow_i(\fixedPointPrice))}}\lr{y_i-\fixedPointFlow_i(\fixedPointPrice)} \geq 0. 
\end{align}
Note that using \eqref{eq: PriceBeta} in \eqref{eq: SelfishCostFirst} we obtain that for any \(y\in F\): 
\begin{align}\label{eq: SocialCostSecond}
    \sum_{i\in[\numLinks]}\lr{\costEdge_i(\bar{x}) + \bar{x}_i \frac{d \costEdge_i(\bar{x}_i)}{dx} + \frac{1}{\beta} \lr{1+\log(\bar{x}_i)}}\lr{y_i-\bar{x}_i} \geq 0, 
\end{align}
where we have used \(\bar{x} = \fixedPointFlow(\fixedPointPrice)\) to simplify the expression. Comparing expression \eqref{eq: SocialCostFirst} and \eqref{eq: SocialCostSecond} we conclude that \(x^\star =\fixedPointFlow(\fixedPointPrice) \). This concludes the proof. 

\subsection{Proof of Lemma \ref{lemma: prep}}
We prove the claims in the lemma sequentially. 

Proof of \textbf{(C1)}: Using the independence of sequence \(\inflow(n),\outflow_i(n)\) we see that the conditional expection of \(\martingale_i(n+1)\) conditioned on the filtration \(\mc{F}_n\) is zero. That is, 
\begin{align*}
    \avg[\martingale_i(n+1)|\mc{F}_n] = \vectorFieldX_i(\flowDiscrete_i(n),\priceDiscrete_i(n)) \lr{\avg[\inflow(n+1)]-\meanInflow} - \flowDiscrete_i(n)\lr{\avg[\outflow_i(n+1)-\meanOutflow]} = 0.
\end{align*}

Proof of \textbf{(C2)}: Recall that we assume that the incoming loads and outgoing fraction have finite support. That is, for every \(i\in[\numLinks]\) and \(n\) we have \(\inflow(n)\in[\inflowLower,\inflowUpper] \subset (0,\infty)\) and \(\outflow_i(n)\in[\outflowLower, \outflowUpper] \subset (0,1)\). Note that \(X_i(0)\geq 0\) for all \(i\in [\numLinks]\) as it is impractical to have negative load on the network; therefore our update reflects the physical constraints of the network. Under this assumption, the discrete-time stochastic update \eqref{eq:x_original} can be written as 
\begin{align*}
    \flowDiscrete_i(n+1) &= (1-\outflow(n+1))\flowDiscrete_i(n) +  \frac{\exp(-\sueParameter \costPrice_i(\flowDiscrete_i(n), \priceDiscrete_i(n)))}{\sum_{j\in[R]}\exp(-\sueParameter \costPrice_j(\flowDiscrete_j(n), \priceDiscrete_j(n)))}\inflow(n),  \\
    &\leq (1-\outflowLower)\flowDiscrete_i(n)+ \inflow(n), \\ 
    &\leq (1-\outflowLower)\flowDiscrete_i(n) + \inflowUpper, \\
    &\leq (1-\outflowLower)^n\flowDiscrete_i(0)+\frac{\inflowUpper}{\outflowLower}.
\end{align*}
The preceding inequalities establish that \(\avg[\|\flowDiscrete_i(n)\|^2]<+\infty\). Consequently, it also ensure that \(\avg\left[ X_i(n)\frac{d \latency_i(n)}{dx}\right] < +\infty\) and is independent of \(n\), which in turn, from \eqref{eq: PricingDynamics} leads to \(\avg[\|P_i(n)\|^2] < +\infty\). 

Proof of \textbf{(C3)}: We see that:
\begin{align*}
    \avg[|\martingale_i(n+1)|^2|\filtration_n] & = \avg[|\vectorFieldX_i(\flowDiscrete_i(n), \priceDiscrete_i(n))(\inflow(n+1) - \meanInflow) - \flowDiscrete_i(n)\lr{\outflow_i(n+1)- \meanOutflow} |^2|\filtration_n] \\ & \leq \avg[ 2 |\vectorFieldX_i(\flowDiscrete_i(n), \priceDiscrete_i(n))(\inflow(n+1) - \meanInflow)|^2 + 2|\flowDiscrete_i(n)\lr{\outflow_i(n+1)- \meanOutflow} |^2|\filtration_n] \\ & \leq \avg[ 2 \frac{\meanInflow}{\meanOutflow}|(\inflow(n+1) - \meanInflow)|^2 + 2|\flowDiscrete_i(n)\lr{\outflow_i(n+1)- \meanOutflow} |^2|\filtration_n] \\ & \leq K_1 + K_2 \|\flowDiscrete_i(n)\|^2 \\ & \leq K\lr{1+|\flowDiscrete_i(n)|^2} <+\infty
\end{align*}
where $K_1 = \frac{2\meanInflow}{\meanOutflow}|\inflowUpper-\inflowLower|^2 $, $K_2 = 2|\outflowUpper-\outflowLower|^2$ and $K = \max(K_1, K_2)$. To obtain the preceding bound we used $\vectorFieldX_i(\flowDiscrete_i(n), \priceDiscrete_i(n)) \leq \frac{\meanInflow}{\meanOutflow}$.  

Proof of \textbf{(C4)}: To prove this result we shall use Theorem \ref{thm: HirshResult}. For that purpose we need to check conditions Theorem \ref{thm: HirshResult}\ref{enum: Assm1}-\ref{enum: Assm4}.

The condition \ref{enum: Assm1} is satisfied if we show that \(\frac{\partial \vectorFieldX_i(\flowContinuous,\priceContinuous)}{\partial \flowContinuous_j} \geq 0\). Indeed, by definition 
\begin{align*}
    \frac{\partial h_i(x,p)}{\partial x_j}  = \frac{\meanInflow}{\meanOutflow} \exp(-\beta \costPrice_i(x_i,p_i)) \frac{\beta\exp(-\beta \costPrice_j(x_j,p_j) )\frac{\partial \costPrice_j(x_j,p_j)}{\partial x_j}}{\lr{\sum_{j\in[\numLinks]}\exp\lr{-\beta \costPrice_j(x_j,p_j)}}^2} > 0.
\end{align*}

Furthermore, note that condition \ref{enum: Assm2} is also satisfied because the Jacobian matrix has all the off-diagonal terms positive and is therefore irreducible. Moreover, it follows from \eqref{eq: hVectorFieldDef} that  for any \(x,p\in\R^{\numLinks}\) , \(\vectorFieldX_i(x,p) \in \lr{0,\frac{\meanInflow}{\meanOutflow}}\) and therefore \ref{enum: Assm3}-\ref{enum: Assm4} are also satisfied. This completes the proof.

Proof of \textbf{(C5)}: We shall now show that \eqref{eq: PriceDynamicsTilde} satisfies conditions \ref{enum: Assm1}-\ref{enum: Assm4} in Theorem \ref{thm: HirshResult}. This ensures global convergence to the equilibrium set which is guaranteed to be singleton by Lemma \ref{lemma:unique_price}. 
For better presentation, for any \(i,j\in[\numLinks]\), we denote the derivative of \(\fixedPointFlow_i(\priceContinuous)\) with respect to \(\priceContinuous_j\) by \(\grad_j \fixedPointFlow(\priceContinuous)\)\footnote{
Note that in Lemma \ref{lemma:monotonicity}, we established that \(\fixedPointFlow(p)\) is continuously differentiable}.  
Using \eqref{eq: xEquilibrium} for any \(i\neq j\) we have
\begin{align}
    \exp(-\beta \costPrice_j(\fixedPointFlow_j(\priceContinuous),\priceContinuous_j)) \fixedPointFlow_i(\priceContinuous)= \exp(-\beta \costPrice_i(\fixedPointFlow_i(\priceContinuous),\priceContinuous_i)) \fixedPointFlow_j(\priceContinuous)
\end{align}
Taking derivative of the above equation with respect to \(\priceContinuous_k\) for \(k\neq i\neq j\): 
\begin{align*}
    & \exp(-\beta \costPrice_j(\fixedPointFlow_j(\priceContinuous),\priceContinuous_j))\grad_k\fixedPointFlow_i(\priceContinuous) \\&\quad + \fixedPointFlow_i(\priceContinuous)\exp(-\beta \costPrice_j(\fixedPointFlow_j(\priceContinuous),\priceContinuous_j))(-\beta \grad_x\costPrice_j(\fixedPointFlow_j(\priceContinuous),\priceContinuous_j) \grad_k \fixedPointFlow_j(\priceContinuous)) \\ 
     &= \exp(-\beta \costPrice_i(\fixedPointFlow_i(\priceContinuous),\priceContinuous_i))\grad_k\fixedPointFlow_j(\priceContinuous) \\ &\quad  + \fixedPointFlow_j(\priceContinuous)\exp(-\beta \costPrice_i(\fixedPointFlow_i(\priceContinuous),\priceContinuous_i))(-\beta \grad_x\costPrice_i(\fixedPointFlow_i(\priceContinuous),\priceContinuous_i) \grad_k \fixedPointFlow_i(\priceContinuous)).
\end{align*}
Collecting similar terms together we get 
\begin{align*}
    &\grad_k\fixedPointFlow_i(\priceContinuous)\lr{\exp(-\beta \costPrice_j(\fixedPointFlow_j(\priceContinuous),\priceContinuous_j))} \\&\quad + \grad_k\fixedPointFlow_i(\priceContinuous)\lr{\fixedPointFlow_j(\priceContinuous)\exp(-\beta \costPrice_i(\fixedPointFlow_i(\priceContinuous),\priceContinuous_i))(\beta \grad_x\costPrice_i(\fixedPointFlow_i(\priceContinuous),\priceContinuous_i) )} \\
    &= \grad_k\fixedPointFlow_j(\priceContinuous) \lr{\exp(-\beta \costPrice_i(\fixedPointFlow_i(\priceContinuous),\priceContinuous_i))} \\&\quad + \grad_k\fixedPointFlow_j(\priceContinuous) \lr{ \fixedPointFlow_i(\priceContinuous)\exp(-\beta \costPrice_j(\fixedPointFlow_j(\priceContinuous),\priceContinuous_j))(\beta \grad_x\costPrice_j(\fixedPointFlow_j(\priceContinuous),\priceContinuous_j) )} .
\end{align*}
This implies for \(i\neq j\neq k\) and for any \(p\in\R^{\numLinks}\) we have
\begin{align}\label{eq: SameSign}
\grad_k \fixedPointFlow_i(\priceContinuous) \cdot \grad_k \fixedPointFlow_j(\priceContinuous) > 0.
\end{align}
Moreover, by definition of fixed point in \eqref{eq: xEquilibrium} we have the constraint that
\[
\sum_{l\in[\numLinks]}\fixedPointFlow_l(p) = \frac{\lambda}{\mu}.
\]
Taking derivative with respect to \(\priceContinuous_k\) of the above equation we obtain
\begin{align}\label{eq: DerivativeConstraint}
\sum_{l\neq k} \grad_k\fixedPointFlow_l(p) = - \grad_k\fixedPointFlow_k(p) > 0,
\end{align}
where the last inequality follows from Lemma \ref{lemma:monotonicity}. Equation \eqref{eq: DerivativeConstraint} in conjunction with \eqref{eq: SameSign} implies that \(\grad_k \fixedPointFlow_i(p) > 0\) forall \(i\neq k\). This ensures satisfaction of \ref{enum: Assm1}-\ref{enum: Assm2}. The requirement \ref{enum: Assm3}-\ref{enum: Assm4} is also satisfied as for any \(p\in\R^{\numLinks}\), \(\fixedPointFlow(\priceContinuous)\in \lr{0,\frac{\meanInflow}{\meanOutflow}}\). This completes the proof. 
\subsection{Proof of Lemma \ref{lemma: borkar}}
Note that to invoke the results from two timescale stochastic approximation theory \cite{borkar2009stochastic}, in addition to Lemma \ref{lemma: prep} we also need to ensure the following 
\begin{enumerate}[label= (\roman*)]
    \item \label{enum: additional1}the function \(\fixedPointFlow(p)\) is Lipschitz;
    \item \label{enum: additional2}the function \(g(x,p) \defas h(x,p)-x\), which is the vector field in 
 \eqref{eq: FixedPointFlowDynamics},    is Lipschitz; 
    \item \label{enum: additional3}the function \(r_i(p) \defas \fixedPointFlow_i(p)\frac{d\ell_i(\fixedPointFlow_i(p))}{dx} - p_i \), which is the vector field in \eqref{eq: PriceDynamicsTilde}, is Lipschitz for all \(i\in[\numLinks]\).   
\end{enumerate}
If the above conditions hold, then \cite[Chapter 9]{borkar2009stochastic} ensures that Theorem \ref{theorem: fixed_point} hold.

Note that \ref{enum: additional1} holds due the fact that \(\fixedPointFlow(p) \in (0,\lambda/\mu)\) and Lemma \ref{lemma:monotonicity} where we established that it is continuously differentiable. Moreover, \ref{enum: additional2} holds due to Lemma \ref{lem: LipschitzFunction}. At last, to show \ref{enum: additional3} we note that for any \(p,p'\in\R^{\numLinks}\) and \(i\in[\numLinks]\):
\begin{align*}
     &\|r_i(p)-r_i(p')\| =
     \big\|\fixedPointFlow_i(p)\frac{d\ell_i(\fixedPointFlow_i(p))}{dx} - \fixedPointFlow_i(p')\frac{d\ell_i(\fixedPointFlow_i(p'))}{dx} \big\| \\
    &\leq \big\|\fixedPointFlow_i(p)\frac{d\ell_i(\fixedPointFlow_i(p))}{dx} - \fixedPointFlow_i(p')\frac{d\ell_i(\fixedPointFlow_i(p))}{dx} \big\| + \big\|\fixedPointFlow_i(p')\frac{d\ell_i(\fixedPointFlow_i(p))}{dx} - \fixedPointFlow_i(p')\frac{d\ell_i(\fixedPointFlow_i(p'))}{dx} \big\| \\ 
    &\leq \big|\frac{d\ell_i(\fixedPointFlow_i(p))}{dx} \big| \|\fixedPointFlow_i(p)-\fixedPointFlow_i(p')\| + |\fixedPointFlow_i(p')|\big\|  \frac{d\ell_i(\fixedPointFlow_i(p))}{dx} - \frac{d\ell_i(\fixedPointFlow_i(p'))}{dx}\big\| \\ 
    &\leq K_1 \bar{L}\|p-p'\|+ K_2 \tilde{L}\|p-p'\|,
\end{align*}
where \(K_1 = \max_{i\in[\numLinks],x\in [0,\lambda/\mu]}\big| \frac{d\ell_i(x)}{dx}\big|\), \(\bar{L}\) is the Lipschitz constant for \(\fixedPointFlow(\cdot)\), \(K_2 = \lambda/\mu  \bar{L}\) and \(\tilde{L}\) is the Lipschitz constant for \(\frac{d \ell_i(x)}{dx}\) when \(x\in [0,\lambda/\mu]\).  

\bibliography{refs}

\begin{thebibliography}{10}

\bibitem{borkar2009cooperative}
Vivek~S Borkar.
\newblock Cooperative dynamics and wardrop equilibria.
\newblock {\em Systems \& control letters}, 58(2):91--93, 2009.

\bibitem{borkar2009stochastic}
Vivek~S Borkar.
\newblock {\em Stochastic approximation: a dynamical systems viewpoint},
  volume~48.
\newblock Springer, 2009.

\bibitem{boyd2004convex}
Stephen~P Boyd and Lieven Vandenberghe.
\newblock {\em Convex optimization}.
\newblock Cambridge university press, 2004.

\bibitem{christodoulou2005price}
George Christodoulou and Elias Koutsoupias.
\newblock The price of anarchy of finite congestion games.
\newblock In {\em Proceedings of the thirty-seventh annual ACM symposium on
  Theory of computing}, pages 67--73, 2005.

\bibitem{cominetti2012modern}
Roberto Cominetti, Francisco Facchinei, and Jean~B Lasserre.
\newblock {\em Modern Optimization Modelling Techniques}.
\newblock Springer Science \& Business Media, 2012.

\bibitem{cominetti2010payoff}
Roberto Cominetti, Emerson Melo, and Sylvain Sorin.
\newblock A payoff-based learning procedure and its application to traffic
  games.
\newblock {\em Games and Economic Behavior}, 70(1):71--83, 2010.

\bibitem{como2021distributed}
Giacomo Como and Rosario Maggistro.
\newblock Distributed dynamic pricing of multiscale transportation networks.
\newblock {\em IEEE Transactions on Automatic Control}, 2021.

\bibitem{farokhi2015piecewise}
Farhad Farokhi and Karl~H Johansson.
\newblock A piecewise-constant congestion taxing policy for repeated routing
  games.
\newblock {\em Transportation Research Part B: Methodological}, 78:123--143,
  2015.

\bibitem{gentle2007matrix}
J.~E. Gentle.
\newblock {\em Matrix {A}lgebra}, volume~10.
\newblock Springer, 2007.

\bibitem{hirsch1985systems}
Morris~W Hirsch.
\newblock Systems of differential equations that are competitive or cooperative
  ii: Convergence almost everywhere.
\newblock {\em SIAM Journal on Mathematical Analysis}, 16(3):423--439, 1985.

\bibitem{nyt2021}
Soumya Karlamangla.
\newblock Why evening rush hour feels so much worse now.
\newblock {\em New York Times}, 2021.
\newblock Published on August 11, 2021. Available at
  \url{https://www.nytimes.com/2021/08/11/us/ca-rush-hour-traffic.html}.

\bibitem{kleinberg2009multiplicative}
Robert Kleinberg, Georgios Piliouras, and {\'E}va Tardos.
\newblock Multiplicative updates outperform generic no-regret learning in
  congestion games.
\newblock In {\em Proceedings of the forty-first annual ACM symposium on Theory
  of computing}, pages 533--542, 2009.

\bibitem{krichene2014convergence}
Walid Krichene, Benjamin Drighes, and Alexandre Bayen.
\newblock On the convergence of no-regret learning in selfish routing.
\newblock In {\em International Conference on Machine Learning}, pages
  163--171. PMLR, 2014.

\bibitem{meigs2017learning}
Emily Meigs, Francesca Parise, and Asuman Ozdaglar.
\newblock Learning dynamics in stochastic routing games.
\newblock In {\em 2017 55th Annual Allerton Conference on Communication,
  Control, and Computing (Allerton)}, pages 259--266. IEEE, 2017.

\bibitem{paccagnan2019incentivizing}
Dario Paccagnan, Rahul Chandan, Bryce~L Ferguson, and Jason~R Marden.
\newblock Incentivizing efficient use of shared infrastructure: Optimal tolls
  in congestion games.
\newblock {\em arXiv preprint arXiv:1911.09806}, 2019.

\bibitem{poveda2017class}
Jorge~I Poveda, Philip~N Brown, Jason~R Marden, and Andrew~R Teel.
\newblock A class of distributed adaptive pricing mechanisms for societal
  systems with limited information.
\newblock In {\em 2017 IEEE 56th Annual Conference on Decision and Control
  (CDC)}, pages 1490--1495. IEEE, 2017.

\bibitem{roughgarden2010algorithmic}
Tim Roughgarden.
\newblock Algorithmic game theory.
\newblock {\em Communications of the ACM}, 53(7):78--86, 2010.

\bibitem{roughgarden2002bad}
Tim Roughgarden and {\'E}va Tardos.
\newblock How bad is selfish routing?
\newblock {\em Journal of the ACM (JACM)}, 49(2):236--259, 2002.

\bibitem{sastry2013nonlinear}
Shankar Sastry.
\newblock {\em Nonlinear systems: analysis, stability, and control}, volume~10.
\newblock Springer Science \& Business Media, 2013.

\bibitem{latimes2021}
Rachel Schnalzer.
\newblock Traffic is terrible again. here's how to get it closer to spring 2020
  levels.
\newblock {\em LA Times}, 2021.
\newblock Published on July 22, 2021. Available at
  \url{https://www.latimes.com/business/story/2021-07-22/los-angeles-traffic-congestion-commute-pandemic}.

\bibitem{sharon2017real}
Guni Sharon, Josiah~P Hanna, Tarun Rambha, Michael~W Levin, Michael Albert,
  Stephen~D Boyles, and Peter Stone.
\newblock Real-time adaptive tolling scheme for optimized social welfare in
  traffic networks.
\newblock In {\em Proceedings of the 16th International Conference on
  Autonomous Agents and Multiagent Systems (AAMAS-2017)}, 2017.

\bibitem{wu2020bayesian}
Manxi Wu, Saurabh Amin, and Asuman Ozdaglar.
\newblock Bayesian learning with adaptive load allocation strategies.
\newblock In {\em Learning for Dynamics and Control}, pages 561--570. PMLR,
  2020.

\end{thebibliography}
\bibliographystyle{plain}

\end{document}